\documentclass[11pt]{article}

\usepackage[margin=1in]{geometry}

\usepackage{amsmath,amssymb,amsthm}
\usepackage{bm}
\usepackage{graphicx}
\usepackage{hyperref}
\usepackage{algorithmicx}
\usepackage{algpseudocode}

\newtheorem{theorem}{Theorem}

\newtheorem{corollary}{Corollary}

\newcommand{\R}{\mathbb{R}}

\newcommand{\Cov}{\mathrm{Cov}}

\newcommand{\matern}{Mat\'ern}

\newcommand{\wt}{\widetilde}

\DeclareMathOperator*{\argmin}{arg\,min}
\DeclareMathOperator*{\argmax}{arg\,max}

\usepackage{setspace}

\usepackage[round]{natbib}

\begin{document}
\vspace*{24pt}

\begin{center}
{\LARGE Permutation and Grouping Methods for}
\vspace{6pt}

{\LARGE Sharpening Gaussian Process Approximations}
\vskip18pt

{\large Joseph Guinness}
\vskip18pt

{\large \textit{North Carolina State University, Department of Statistics}}
\end{center}
\vspace{2pt}

\abstract{Vecchia's approximate likelihood for Gaussian process parameters depends on how the observations are ordered, which has been cited as a deficiency. This article takes the alternative standpoint that the ordering can be tuned to sharpen the approximations. Indeed, the first part of the paper includes a systematic study of how ordering affects the accuracy of Vecchia's approximation. We demonstrate the surprising result that random orderings can give dramatically sharper approximations than default coordinate-based orderings. Additional ordering schemes are described and analyzed numerically, including orderings capable of improving on random orderings. The second contribution of this paper is a new automatic method for grouping calculations of components of the approximation. The grouping methods simultaneously improve approximation accuracy and reduce computational burden. In common settings, reordering combined with grouping reduces Kullback-Leibler divergence from the target model by more than a factor of 60 compared to ungrouped approximations with default ordering. The claims are supported by theory and numerical results with comparisons to other approximations, including tapered covariances and stochastic partial differential equations. Computational details are provided, including the use of the approximations for prediction and conditional simulation. An application to space-time satellite data is presented.}

\section{Introduction}

The Gaussian process model has become very popular for the analysis of time series, functional data, spatial data, spatial-temporal data, and the output from computer experiments. Likelihood-based methods for estimating Gaussian process covariance parameters were popularized by the work of \cite{mardia1984maximum}, but it was quickly realized that the $O(n^2)$ memory and $O(n^3)$ flop burden made these methods infeasible for large datasets, necessitating the use of computationally efficient approximations. This article studies one of the earliest such approximations, due to \cite{vecchia1988estimation}, which has a number of advantages. First, the approximation is embarassingly parallel, making it a good candidate for implementation on high performance computing systems. Second, the approximation corresponds to a valid multivariate normal distribution, meaning that the approximation can be used to generate ensembles of conditional simulations to characterize joint uncertainties in predictions. Third, the approximation enjoys the desirable statistical property that maximizing it corresponds to solving a set of unbiased estimating equations \citep{stein2004approximating}. Lastly, the approximation is very accurate; we demonstrate that with the improvements described in this article, it far outperforms state-of-the-art methods such as stochastic partial differential equation (SPDE) approximations \cite{lindgren2011explicit} and covariance tapering \citep{furrer2006covariance,kaufman2008covariance}.

 Suppose that the $n$ observation locations are labeled 1 through $n$, as in $x_1,\ldots,x_n \in \R^d$. Define $y_i := y(x_i) \in \R$ to be the observation associated with location $x_i$ and the vector $y = (y_1,\ldots,y_n)$. In a Gaussian process, the data $y$ are modeled as $Y \sim N(\mu,\Sigma_\theta)$, where $\Sigma_\theta$ is a covariance matrix with $(i,j)$ entry determined by covariance function $K_\theta(x_i,x_j)$ depending on a vector of covariance parameters $\theta$. Let $\tau: \{1,\ldots,n\} \to \{1,\ldots,n\}$ be a permutation of the integers 1 through $n$, and define the permuted vector $y^\tau$, where $y^\tau_i = y_{\tau(i)}$. For any permutation $\tau$, the joint density for the observations can be written as a product of conditional densities
\begin{align*}
p_\theta(y_1,\ldots,y_n) = p_\theta(y^\tau_1) \prod_{i=2}^n p_\theta(y^\tau_i|y^\tau_1,\ldots,y^\tau_{i-1}),
\end{align*}
reflecting the invariance of the joint density to arbitrary relabeling of the observations. Vecchia's approximation replaces the complete conditioning vectors $(y^\tau_1,\ldots,y^\tau_{i-1})$, with a subvector. Specifically, let $\{j_{i1},\ldots,j_{im_i}\}$ be a set of integers between $1$ and $i-1$, and define the approximation
\begin{align}\label{vecchiaapprox}
p_{\theta,\tau,J}(y_1,\ldots,y_n) = p_\theta(y^\tau_1) \prod_{i=2}^n p_\theta(y^\tau_i|y^\tau_{j_{i1}},\ldots,y^\tau_{j_{im_i}}),
\end{align}
where $J = \{J_1,\ldots,J_n\}$ is collection of sets $J_i := \{j_{i1},\ldots,j_{im_i},i\}$, which we refer to as the neighbors of observation $i$ (by convention every observation neighbors itself).  Since Vecchia's approximation is defined as an ordered sequence of valid conditional distributions, $p_{\theta,\tau,J}$ forms a valid joint distribution. Also, its definition as a product of conditional distributions allows the approximation to be easily parallelized.

Clearly, the quality of the approximation depends on the choice of $J$ since the approximate likelihood differs from the exact likelihood if any $m_i \neq i-1$. The quality of the approximation also depends on the permutation $\tau$, a fact acknowledged in the literature, but not yet carefully explored until now. \cite{banerjee2014hierarchical} outline some of the common criticisms of Vecchia's approximation,
\begin{quotation}
``However, the approach suffers many problems. First, it is not formally defined. Second, it will typically be sequence dependent, though there is no natural ordering of the spatial locations. Most troubling is the arbitrariness in the number of and choice of `neighbors.' Moreover, perhaps counter-intuitively, we can not merely select locations close to $\bm{s}_i$ (as we would have with the full data likelihood) in order to learn about the spatial decay in dependence for the process. So, altogether, we do not see such approximations as useful approach.''
\end{quotation}
This paper takes the stance that viewing order dependence as an aspect of the approximation that can be tuned--rather than as an inherent deficiency--is a fruitful avenue for improving the approximation. \cite{datta2016hierarchical} conducted a simulation study to compare three coordinate-based orderings and concluded that the inferences are ``extremely robust to the ordering of the locations.'' On the contrary, default orderings based on sorting the locations on a coordinate are often badly suboptimal compared to even a completely random ordering of the points. On two-dimensional domains, the more carefully constructed maximum minimum distance ordering is a further improvement and can achieve greater than 99\% relative efficiency for estimating covariance parameters with as few as 30 neighbors, chosen to be simply the 30 nearest previous points. These results address most of the concerns raised by \cite{banerjee2014hierarchical}. Order dependence can be exploited to improve the approximation. For maximum minimum distance ordering, a simple choice of nearest neighbors is effective for this ordering. We prove that the quality of the model approximation is nondecreasing as the number of neighbors increases, so that the choice of the number of neighbors is governed by a natural tradeoff between computational efficiency and model approximation.

In addition to the results on orderings, this paper describes how the approximation can be computed more efficiently, both in terms of memory burden and floating point operations, by grouping the observations into blocks and evaluating each group's contribution to the likelihood simultaneously. The grouped version of the approximation is particularly interesting because not only does it reduce the computational burden, but it is provably guaranteed to improve the model approximation. A computationally efficient algorithm for grouping the observations is described and implemented, and numerical studies supporting the theoretical results are presented.

The consideration of arbitrary orderings presents some computatational issues that are not problematic for coordinate-based orderings. First, obviously, is the issue of how to find the orderings. The maximum minimum distance ordering requires $O(n^3)$ floating point operations. To address this issue, we introduce an $O(n\log n)$ algorithm for finding an ordering that mimics the salient features of the maximum minimum distance ordering. Second, a naive search for ordered nearest neighbors requires $O(n^2 \log n)$ flops. There are well-known computationally efficient methods for finding nearest neighbors, and this paper describes an adaptation of those methods to the case when the neighbors must come from earlier in the ordering. We also describe methods for profiling out linear mean parameters and using the approximation for spatial prediction and for efficiently drawing from the conditional distribution at a set of unobserved locations given the data, which is a useful way to quantify joint uncertainty in predicted values.

Several articles have investigated properties and extensions of Vecchia's approximation. \cite{pardo1997amle3d} describe Fortran 77 software for computing the approximate likelihood and compare a random ordering versus a sorted coordinate ordering on a small dataset of size 41. \cite{stein2004approximating} describe how Vecchia's approximation can be extended to residual/restricted maximum likelihood estimation, how blocking can be used to speed the computations, and how conditioning on nearest neighbors can be suboptimal when points are ordered according to the coordinates. \cite{sun2016statistically} use Vecchia's approximation to define several unbiased estimating equations for covariance parameters. \cite{datta2016hierarchical} use Vecchia's approximation as part of a hierarchical Bayesian specification of spatial processes, and \cite{datta2016nearest} discuss interpretation of the approximation as an approximation to the inverse Cholesky factor of the covariance matrix and apply it to multivariate spatial data. Vecchia's approximation has been used in various applications, including for seismic data \citep{eide2002seismic} and space-time SPDE models \citep{jones1997models}.

The paper is organized as follows. Section \ref{permutationdesigns} outlines formal definitions for the orderings. Section \ref{grouping} presents the computational and theoretical results related to grouping the observations, establishing that grouping observations simultaneously reduces computational effort and improves the model approximations. Section \ref{furthercomputational} addreses additional computational issues described above. Section \ref{numerical} contains numerical and timing experiments studying the relative efficiencies for various orderings, Kullback-Leibler divergence for various orderings and for other proposed Gaussian process approximations, and timing results. Section \ref{datasection} includes an application of the methods to space-time satellite data, and the paper concludes with a discussion in Section \ref{discussion}.

\section{Definitions of Orderings}\label{permutationdesigns}

In the numerical analysis literature on sparse matrix factorizations, it is widely recognized that row-column reordering schemes for large sparse symmetric positive definite matrices are essential for increasing the sparsity of the Cholesky factor \citep{saad2003iterative}. Finding an optimal such ordering is an NP-complete problem \citep{yannakakis1981computing}, and so the algorithms in use are necessarily heuristic, but heuristic algorithms, such as the approximate minimum degree algorithm \citep{amestoy1996approximate}, have proven to be effective. The goals in Gaussian process approximation are related in that we search for an ordering that produces an approximately sparse inverse Cholesky factor. However, the problem at hand here is more difficult and perhaps less well-defined than sparse matrix factorization; our task is to find the best reordering of observations that produces accurate approximations to the joint density or an approximate likelihood function that delivers efficient parameter estimates, a criterion that depends on the derivative of the approximate and exact likelihood functions with respect to the covariance parameters. Thus, it is extremely unlikely that we will be able to find ``optimal'' orderings for large datasets. Nevertheless, as in the sparse matrix case, this paper shows that heuristically motivated orderings can offer significant improvements in statistical efficiency and model approximation over default orderings.

It appears that the default choice for Vecchia's approximation is to order the points by sorting on one of the coordinates. This is the ordering used in \cite{vecchia1988estimation}, \cite{sun2016statistically}, and \cite{datta2016hierarchical}. \cite{stein2004approximating} use an ordering based on sorting the points on the sum of their coordinates, which is equivalent to ordering on one of the coordinates in a system rotated by $\pi/4$. We refer to such orderings as \textit{sorted coordinate} orderings. Sorted coordinate orderings are based on a heuristic from one-dimensional examples that each location can be separated by previous locations in the ordering by its nearest neighbors.

The numerical studies in Section \ref{numerical} indicate that the following ordering scheme is effective for Mat\'ern covariance models in two dimensions. This ordering selects a point in the center first--the center being the mean location or some other measure of centrality, generically denoted as $\overline{x}$--then sequentially picks the next point to have maximum minimum distance to all previously selected points, that is
\begin{align*}
\tau(1) &= \argmin_{i \in 1,\ldots, n} \| x_i - \overline{x} \|, \\
\tau(j) &= \argmax_{i \notin \tau(1),\ldots,\tau(j-1) } \min_{ k \in 1,\ldots,j-1 } \| x_i -  x_{\tau(k)} \|, \quad  j > 1 .
\end{align*}
The result is that for every $k=1,\ldots,n$, the first $k$ points form a space-covering set, none of which are too near each other. We refer to this ordering as a \textit{maximum minimum distance} (MMD) ordering, and any approximation to it as an approximate MMD (AMMD) ordering. The MMD ordering is based on a heuristic of making sure that each location is surrounded by previous locations in the ordering.

The numerical studies in more than two dimensions indicate that it can be beneficial to sort the points by their distance to some point in the domain. For example, the points can be sorted based on their distance to $\overline{x}$. We use \textit{middle out} ordering to refer to sorting based on distance to the center. The heuristic for middle out is similar to the sorted coordinate heuristic, with the difference that previous points fall inside a sphere with radius smaller than the radius of the current point. Finally, a \textit{completely random} ordering is a draw from the uniform distribution on all permutations. Random orderings do not have a heuristic but tend to give orderings with the same surrounding heuristic as MMD and outperform sorted coordinate orderings in many cases. Examples of the four orderings are given in Figure \ref{orderplots}.

\begin{figure}
\centering
\includegraphics[width=0.8\textwidth]{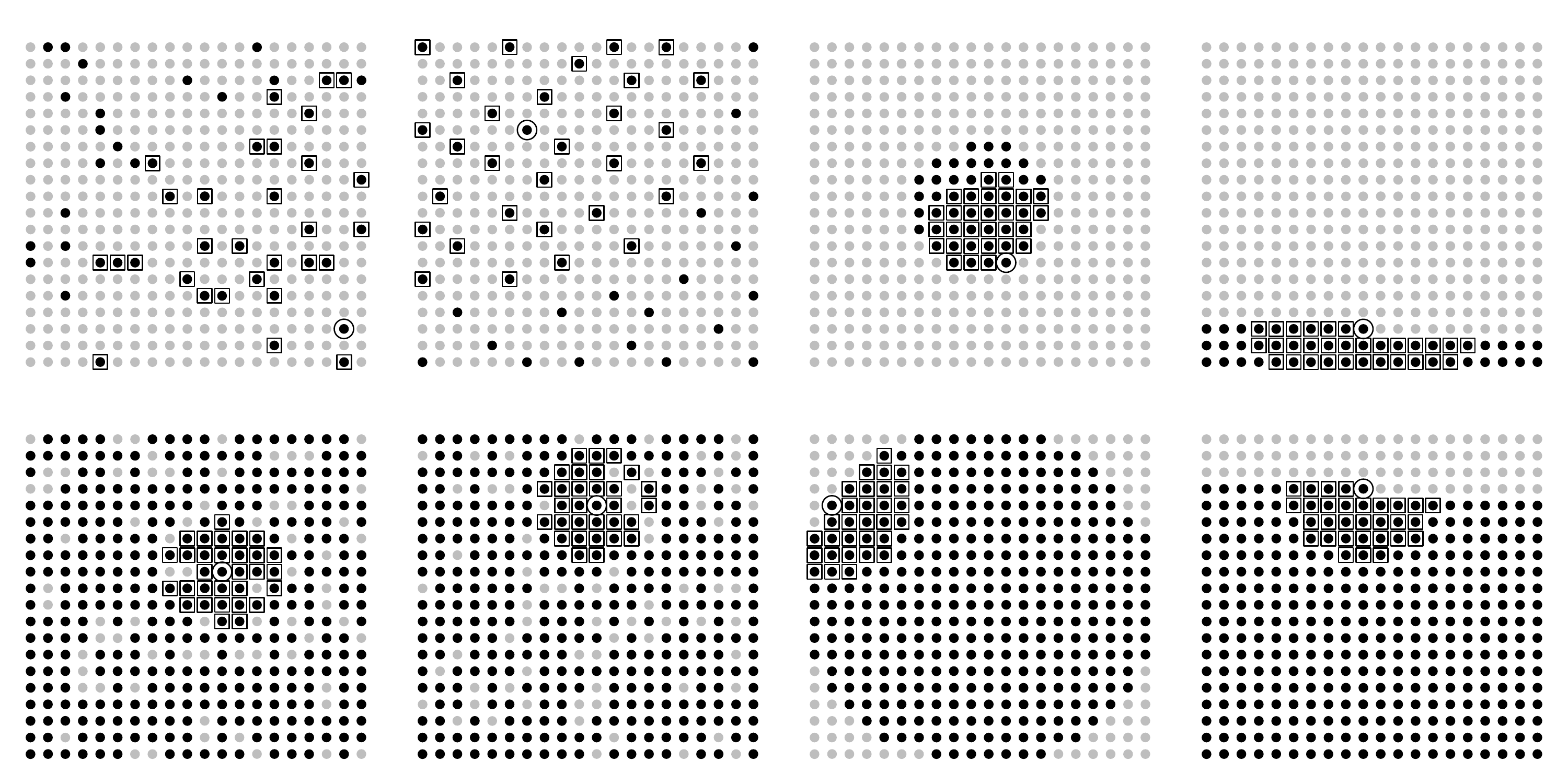}
\caption{{\small Examples of four orderings of 400 locations on a $20 \times 20$ grid. From left to right are completely random, maximum minimum distance, middle out, and sorted on vertical coordinate. On the top row, the black circles are ordered 1-50, circled point is number 50, and squared points are 30 nearest neighbors to point 50 among previous points. Bottom row is the same for point ordered number 330.}}
\label{orderplots}
\end{figure}

%%%%%%%%%%%%%%%%%%%%%%%%%%%%%%%%%%%%
%   Grouping
%%%%%%%%%%%%%%%%%%%%%%%%%%%%%%%%%%%

\section{Automatic Grouping Methods}\label{grouping}

In this section, a grouped version of Vecchia's approximation is described and proven to be no worse (and possibly better) than its ungrouped counterpart in terms of Kullback-Leibler (KL) divergence. Moreover, this section outlines how an entire group's contribution to the likelihood can be computed simultaneously, and it is demonstrated that if the grouping is chosen well, the simultaneous computation is lighter, both in terms of floating point operations and memory, which is important for implementation on shared memory parallel systems. Essentially, grouping gives us better approximations for free.

\cite{stein2004approximating} considered a blockwise version of Vecchia's likelihood and demonstrated that blocking can reduce the computational burden. The grouping methods developed here differ in that they operate on the ordering and neighbor sets after they have been defined. This gives us the freedom to first choose an advantangeous ordering and neighbor sets, and then use grouping to reduce the computational effort and improve the model accuracy. The grouping scheme here is also more general in that observations within a group are not required to be ordered continguously, nor are observations within one group required to condition on all observations in another group.

Let $B = \{B_1,\ldots,B_K\}$ be a partition of $\{1,\ldots,n\}$, representing a grouping of the observations into blocks, and define $U_k = \bigcup_{b \in B_k} J_b$,
the union of all neighbors of all indices $B_k$, representing the set of all neighbors of all observations in block $B_k$. Suppose that $i \in B_k$, and define $\overline{J}_i$ as follows:
\begin{align*}
\overline{J}_i = \{j \in U_k : j \leq i\},
\end{align*}
so that $\overline{J}_i$ is the union of all neighbors of the observations with index in $B_k$, subject to the requirement that all elements of $\overline{J}_i$ must be less than or equal to $i$. Thus $\overline{J}_i$ is a set of integers between 1 and $i$ that includes $i$. Therefore $\overline{J}_i$ conforms to our convention and $p_{\theta,\tau,\overline{J}}$ is an approximation in the form of \eqref{vecchiaapprox}. Note that $\overline{J}_i$ depends on the original neighbor choices $J_1,\ldots,J_n$ and the partition $B$. I refer to $p_{\theta,\tau,\overline{J}}$ as the grouped version of the ungrouped approximation $p_{\theta,\tau,J}$. Figure \ref{groupingfigure} gives an example of how $\overline{J}_i$ is constructed for a group containing two points.

\begin{figure}
\centering
\includegraphics[width=0.8\textwidth]{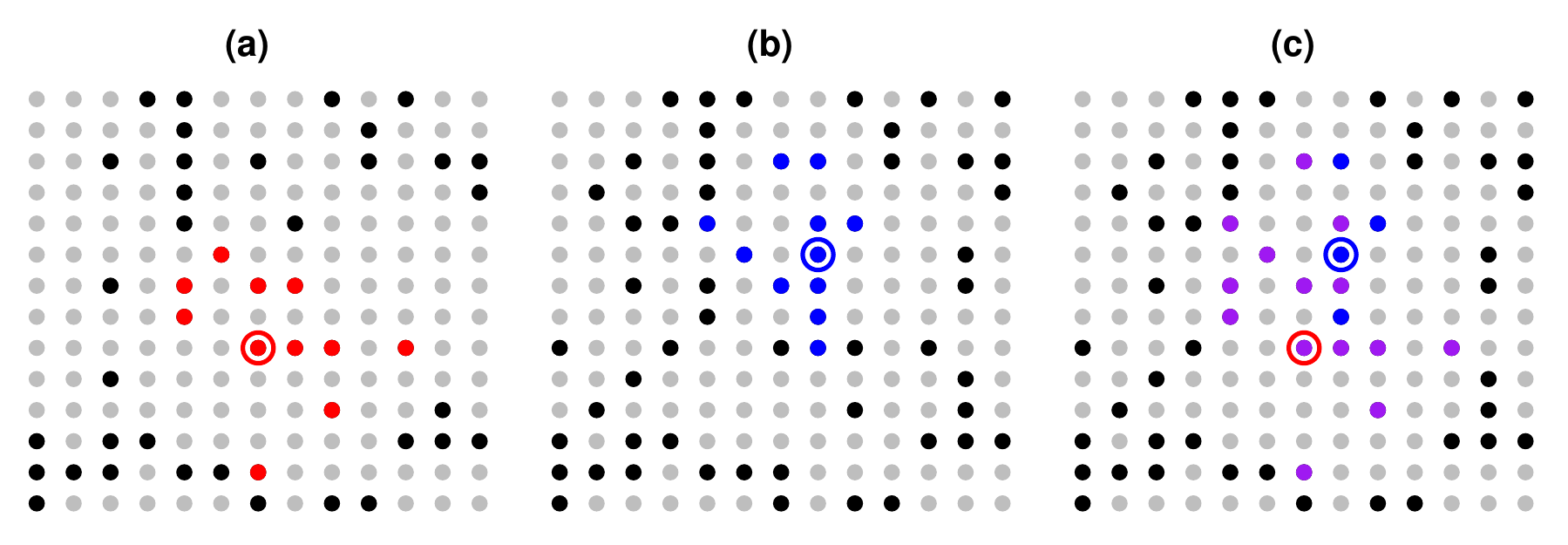}
\caption{{\small Example with $B_1 = \{45,60\}$. (a) Point 45 (circled), points 1 through 45 (black or red), $J_{45}$ (red). (b) Point 60 (circled), points 1 through 60 (black or blue), $J_{60}$ (blue). (c) Points 1 through 60 (not gray), $\overline{J}_{45}$ (purple), $\overline{J}_{60}$ (purple + blue). $\overline{J}_{45}$ has three more points than $J_{45}$, and $\overline{J}_{60}$ has 7 more points than $J_{60}$. }}
\label{groupingfigure}
\end{figure}

\subsection{Improved Model Approximations}

To gain an intuition for why grouping observations improves the accuracy of the model approximation, I briefly review one interpretation of the inverse Cholesky factor of a covariance matrix \citep{pourahmadi1999joint}. Suppose $Y = (Y_1,\ldots,Y_n)^T$ is a mean-zero multivariate normal vector with $E(YY^T) = \Sigma$, which has Cholesky decomposition $\Sigma = LL^T$. Let $\Gamma = L^{-1}$. The transformation $Z = \Gamma Y$ is decorrelating and can be interpreted as follows; for any $i=2,\ldots,n$,
\begin{align*}
Z_i = \sum_{j=1}^i Y_j = \Gamma_{ii} \bigg[ Y_i - \sum_{j=1}^{i-1} \bigg(\frac{-\Gamma_{ij}}{\Gamma_{ii}}\bigg) Y_{j} \bigg],
\end{align*}
is the standardized residual of the projection of $Y_i$ onto $(Y_1,\ldots,Y_{i-1})$, $\Gamma_{ii}$ is the reciprocal of the residual standard deviation, and $-\Gamma_{ij}/\Gamma_{ii}$ are the coefficients of the projection. The coefficients can also be interpreted as those defining the best linear unbiased predictor (BLUP) of $Y_i$ using $(Y_1,\ldots,Y_{i-1})$, an interpretation which can be extended to the case of a linear mean function \citep{stein2004approximating}. Defining $z = \Gamma y$, the conditional density for $Y_i$ given $Y_1,\ldots,Y_{i-1}$ can be written as
\begin{align*}
\log p_\theta(y_i|y_1,\ldots,y_{i-1}) = -\frac{1}{2}\log 2\pi + \log \Gamma_{ii} - \frac{1}{2}z_i^2,
\end{align*}
and thus the $i$th row of $\Gamma$ encodes the conditional density for $Y_i$ given $Y_{1},\ldots,Y_{i-1}$.

Vecchia's approximation replaces $\Gamma$ with a sparse approximation $\wt{\Gamma}$, since $Y_i$ is projected onto a subset of $(Y_1,\ldots,Y_{i-1})$, and thus $\wt{\Gamma}_{ij}$ is zero if $Y_j$ is not in the conditioning set for $Y_i$. Because of the BLUP interpretation, $\wt{\Gamma}_{ii} \leq \Gamma_{ii}$ for every $i$. Taking this a step further, suppose $\wt{\Gamma}^{1}$ and $\wt{\Gamma}^2$ represent two approximations that use the same ordering with different conditioning sets given by $J^1$ and $J^2$. If $J_i^1 $ is contained in $J_i^2$ for every $i$, then $\wt{\Gamma}^{1}_{ii} \leq \wt{\Gamma}^2_{ii}$, because the variance of the BLUP cannot increase when additional variables are added to the conditioning set.

The Kullback-Leibler (KL) divergence from $p_1$ to $p_0$ is defined as $KL(p_0||p_1) = E( \log(p_0/p_1) )$, where the expectation is with respect to $p_0$. The following theorem establishes that the KL divergence is nonincreasing as observations are added to the conditioning sets. This result is needed to show that the grouped approximation is no worse than its ungrouped counterpart (and could be better).

\begin{theorem}\label{KLtheorem}
If $J^1_i \subset J^2_i$ for every $i = 1,\ldots,n$, then $KL(p_\theta || p_{\theta,\tau,J^2}) \leq KL(p_\theta || p_{\theta,\tau,J^1})$.
\end{theorem}

\begin{proof}
Let $\Sigma^1$ and $\Sigma^2$ be the covariance matrices under $p_{\theta,\tau,J^1}$ and $p_{\theta,\tau,J^2}$. The difference between the KL-divergences is given by
\begin{align*}
KL(p_\theta || p_{\theta,\tau,J^1}) - KL(p_\theta || p_{\theta,\tau,J^2}) = \frac{1}{2}\left( E( Y^T (\Sigma^1)^{-1} Y ) - E( Y^T (\Sigma^2)^{-1} Y ) - \log \det (\Sigma^1)^{-1} + \log \det (\Sigma^2)^{-1} \right).
\end{align*}
Let $\Gamma^k$ be the inverse Cholesky factor of $\Sigma^k$. According to Vecchia's approximation, $\Gamma^k$ is sparse and has entries constructed as follows. Let $A$ be the covariance matrix for the vector
\begin{align*}
( Y_{j_{i1}}, \ldots, Y_{j_{im}}, Y_i)^T,
\end{align*}
and let $C$ be the inverse Cholesky factor of $A$. Row $m+1$ (the last row) of $C$ contains the coefficients for the projection of $Y_i$ onto $(Y_{j_{i1}},\ldots,Y_{j_{im}})$. The $i$th row of $\Gamma^k$ has entries
\begin{align*}
\Gamma^k_{ij_{i\ell}} = C_{m+1,\ell}, \, \, \Gamma^k_{ii} = C_{m+1,m+1}, \, \, \Gamma^k_{ij} = 0 \, \, \mbox{otherwise}.
\end{align*}
Define $Z^k = \Gamma^k Y$. Then $Z^k_i$ is
\begin{align*}
Z^k_i = \sum_{j=1}^i \Gamma^k_{ij} Y_j = \sum_{\ell=1}^{m} C_{m+1,\ell} Y_{j_{i\ell}} + C_{m+1,m+1} Y_i = C_{m+1,m+1}\Big[ Y_i - \sum_{\ell=1}^m \Big( \frac{-C_{m+1,\ell}}{C_{m+1,m+1}} \Big) Y_{j_{i\ell}} \Big].
\end{align*}
From the last expression it is evident that $Z^k_i$ is the residual from the BLUP of $Y_i$ given $(Y_{j_{i1}},\ldots,Y_{j_{im}})$, scaled by the standard deviation of that residual, and so $Z^k_i$ has mean zero and variance 1. This ensures that $E( Y^T (\Sigma^k)^{-1} Y ) = E( \sum_i (Z^k_i)^2 ) = n$, and thus both quadratic form terms are equal to $n$ and cancel each other.

The determinant of $\Sigma_k^{-1}$ can be computed by taking the square of the product of the entries of $\Gamma^k$. Since $J^1_i \subset J^2_i$ for every $i$, we have that $0 \leq \Gamma^1_{ii} \leq \Gamma^2_{ii}$ for every $i$, so $\log \det (\Sigma^1)^{-1} \leq \log \det (\Sigma^2)^{-1}$, and thus $KL(p_\theta || p_{\theta,\tau,J^2}) \leq KL(p_\theta || p_{\theta,\tau,J^1})$, as desired.
\end{proof}

Corollary 1 follows from Theorem 1 when we note that, by definition, $\overline{J}_i$ contains $J_i$, since $i \in B_k$ and all elements of $J_i$ are less than or equal to $i$.

\begin{corollary}\label{groupingtheorem}
The grouped approximation has smaller KL divergence than its ungrouped counterpart, i.e.\
\begin{align*}
KL(p_\theta || p_{\theta,\tau,\overline{J}}) \leq KL(p_\theta || p_{\theta,\tau,J})
\end{align*}
\end{corollary}

\subsection{Simultaneous Computation}\label{computationaladvantages}

Recall that $U_k$ is the set of all neighbors of observations in block $B_k$. Let $A^{k}$ be the covariance matrix for observations with indices in $U_k$, with observations ordered increasing in their indices. If we define $C^k$ to be the inverse Cholesky factor of $A^k$, each row of $C^k$ encodes the conditional distribution of an observation with index in $U_k$ given those ordered before in $U_k$. Since $\overline{J}_i$ contains precisely those indices ordered before $i$ in $U_k$, we can calculate the contribution of observations with indices $i \in B_k$ to $p_{\theta,\tau,\overline{J}}$ simultaneously, by forming and factoring a single covariance matrix $A^k$. The total computational cost for the grouped version is on the order
\begin{align*}
\mbox{memory: } \sum_{k=1}^K (\# U_k)^2, \quad \mbox{flops: } \sum_{k=1}^K (\# U_k)^3,
\end{align*}
where $(\# U_k)$ is the size of the set $U_k$. If a group of observations share a substantial number of neighbors, it is possible that these quantities are less than the memory and flop burden of the ungrouped approximation--$n(m+1)^2$ and $n(m+1)^3$--and so computational savings can be gained by computing each group's contribution to the likelihood in this simultaneous way.

To illustrate this point, consider the canonical ordering and two observations $y_i$ and $y_{i+2}$ with $J_i = ({i-m},\ldots,i)$ and $J_{i+2} = ({i-m+2},\ldots,{i+2})$, so that each observation conditions on the $m$ previous ones. The computational cost of factoring the covariance matrices for the observations and their respective neighbors is $1/3(m+1)^3$ floating point operatations each, for a total of $2/3(m+1)^3$ operations, and $(m+1)^2$ memory units each, for a total of $2(m+1)^2$ memory units.
Now, suppose that $B_1 = \{i,i+2\}$, and thus $U_1 = \{i-m,\ldots,i+2\}$, $\overline{J}_i = J_i$, $\overline{J}_{i+2} = \{i-m,\ldots,{i+2}\}$, and $A^1$ is the $(m+3) \times (m+3)$ covariance matrix for $(y_{i-m},\ldots,y_{i+2})$. Row $m+1$ of $C^1$ encodes the conditional distribution of $y_i$ given $(y_{i-m},\ldots,y_{i-1})$, and row $m+3$ of $C^1$ encodes the conditional distribution of $y_{i+2}$ given $(y_{i-m},\ldots,y_{i+1})$, which is what is required for calculating the contribution from $y_i$ and $y_{i+2}$ to the grouped likelihood approximation. The computational cost of storing and factoring $A^1$ is $(m+3)^2$ memory units and $1/3(m+3)^3$ flops, so in addition to the increased accuracy of the approximation (Corollary 1), we also achieve computational savings in flops if $m>6$ and in memory if $m>3$.
\vskip12pt

\subsection{Grouping Algorithm}

Since the groupings are a partition of the set of observations, it is of interest to find good partitions that improve the model and lessen computational cost. Here, we describe a fast greedy algorithm that depends only on the ordering and the neighbor sets. The algorithm starts with $B_k = \{k\}$, each index in its own block and proceeds to propose joining pairs of blocks.
\vskip12pt

\noindent \textbf{Grouping algorithm}:
\vskip6pt

\begin{algorithmic}
\State Set $B = \{ \{1\},\ldots,\{n\}\}$
\For{$\ell = 1,\ldots,m$}
	\For{$i = 1,\ldots,n$}

		\State Identify the $k$ for which $i \in B_k$

		\State Identify the $k'$ for which $j_{i\ell} \in B_{k'}$

		\If{$(\#(U_k \cup U_{k'}))^2 \leq (\# U_k )^2 + (\# U_{k'})^2$}

			\State Set $B_k = B_k \cup B_{k'}$

			\State Set $B_{k'} = \{\}$

		\EndIf

	\EndFor
\EndFor
\end{algorithmic}
\vskip12pt

\noindent The joining of two blocks is accepted if the square of the number of neighbors of the joined block is less then the sum of the squares of the number of neighbors in the two blocks. This ensures that the memory burden is made no worse when combining blocks, since the covariance matrices we need to form are governed by $(\# U_k)^2$. In practice, this rule also does not increase the computational burden since, when working with small covariance matrices, the computational demand is often dominated by filling in the entries of the covariance matrix rather than factoring the matrix. Only neighbors are considered as candidates for joining, ensuring that at most $nm$ comparisons are made, and that we make comparisons between blocks that contain points that are near each other, which is advisable because distant points are unlikely to share many neighbors unless they appear early in the ordering.

\section{Further computational considerations}\label{furthercomputational}

This section describes computational considerations for the problems of finding orderings and nearest neighbors, approximately drawing from unconditional and conditional Gaussian process distributions, and profiling out linear mean parameters. Code for all methods is provided in online supplementary material.

\subsection{Finding Orderings and Nearest Neighbors}

Considering arbitrary permutations of the observations presents two computational issues that are not encountered in coordinate-based orderings. The first is finding the orderings. For example, finding the MMD ordering is $O(n^3)$ in computing time. The second problem is how to find the ordered nearest neighbors, that is, the set of $m$ nearest neighbors to location $x_{\tau(i)}$ among $x_{\tau(1)},\ldots,x_{\tau(i-1)}$. A naive algorithm for finding neighbors would require $O(n\log n)$ operations to compute and sort distances to any point, and thus a total of $O(n^2 \log n)$ operations.

Searching for nearest neighbors in metric spaces is a well-studied problem. \cite{vaidya1989ano} describes an $O(n\log n)$ algorithm for finding the nearest neighbor to all $n$ points. The FNN R package \citep{FNN} has an implementation of a kd-tree algorithm. The algorithm uses a tree-based partitioning of the points to quickly narrow down the set of points that are candidates for being nearest each point. However, the FNN software cannot be used directly because the functions return nearest neighbors from the set of all points, rather than nearest previous neighbors, which is what we need for Vecchia's approximation. To bridge this gap, we first find the nearest $2m$ neighbor sets for each point. For the neighbor sets that contain at least $m$ previous points, we have found the nearest $m$ previous points. Those points are set aside, and for the remaining points, we find the nearest $4m$ points, again setting aside the points that have at least $m$ previous neighbors. This is repeated until all points previous neighbors have been found.

It is possible to construct approximations to the MMD ordering that run in %$O(n\log n)$
% time as well; I describe the one used in our examples here. The algorithm begins with a random ordering. Then, for the first $n/8$ points, we use the FNN functions to find the $k$ nearest neighbors. Then we disperse the neighbors of the first point to later in the ordering. This guarantees that none of the $k$ nearest neighbors of the first point will appear early in the ordering. Then we repeat this for the second point in the ordering, the third point, and so on. For the $j$th point in the ordering we disperse only its $k - j/(n/8)$ nearest neighbors.
%$O(n^2)$ or
$O(n\log n)$ time.
%We first describe a quadratic time algorithm. To select the $k$th point in the ordering, randomly choose $M$ points from the set of unselected points, and then select the $k$th point to be the one among the $M$ randomly selected points that has maximum minimum distance to the already selected ponts. This runs in $O(Mn^2)$ time. An $O(n\log n)$ algorithm proceeds as follows.
First divide the domain up into grid boxes, with the number of grid boxes proportional to the number of observations. Assigning each point to its grid box can be done in linear time with a simple rounding method. Using the center of each grid box as its location, then use MMD to order the grid boxes. Once the grid boxes are ordered, loop over the grid boxes according to the grid box ordering, each time picking the point from the current grid box with largest minimum distance to any selected points in the current or neighboring grid boxes, stopping when all points have been ordered. If the number of grid boxes $g$ is large this algorithm can be used recursively to order the grid boxes, thus giving an $O(n\log n)$ algorithm.

\subsection{Unconditional and Conditional Draws}\label{computational}

Whereas Kriging interpolation requires only matrix multiplication and linear solves with the covariance matrix, which can usually be done in $O(n^2)$ operations with a good iterative solver, unconditional and conditional draws (simulations) from the Gaussian process model generally require us to factor the covariance matrix, which is $O(n^3)$ operations. One advantage that Vecchia's approximation has is that it defines an approximation to the inverse Cholesky factor, which can be used to perform approximate draws from the Gaussian process model. This approximation can also be used to perform conditional draws of unobserved values given the observations, which, when done many times in an ensemble, is a useful way of quantifying the joint uncertainty associated with interpolated maps.

To simplify the following equations, write $\Gamma$ as the inverse Cholesky factor of the $n\times n$ covariance matrix $\Sigma$, and $\widetilde{\Gamma}$ as the approximation to $\Gamma$ implied by Vecchia's likelihood. Let $Z$ be an uncorrelated vector of standard normals of length $n$. Then $Y = \widetilde{\Gamma}^{-1} Z$ is approximately $N(0,\Sigma)$. We use Kullback-Leibler divergence in Section \ref{numerical} to monitor the quality of the approximations. Since $\widetilde{\Gamma}$ is sparse and triangular, solving for $Y$, and thus drawing approximately from from $N(0,\Sigma)$, can be done in $O(n)$ operations.
%For models with finite variance, Kriging interpolation is equivalent to conditional expectation in the multivariate normal.
Suppose $Y$ is parititioned into two subvectors as $Y = (Y_1,Y_2)$ and write $\Sigma$ as a $2\times2$ block matrix $\{ \Sigma_{ij} \}_{i=1,2}$. To predict $Y_2$ from $Y_1$, We compute $E(Y_2|Y_1) = \Sigma_{21}\Sigma_{11}^{-1}Y_1$. This can be rewritten in terms of the inverse Cholesky factor as $E(Y_2|Y_1) = -\Gamma_{22}^{-1}\Gamma_{21}Y_1$. Thus we can compute approximate conditional expectations as $-\widetilde{\Gamma}_{22}^{-1}\widetilde{\Gamma}_{21} Y_1$. Solving a linear system with $\widetilde{\Gamma}_{22}$ can also be achieved in linear time because $\widetilde{\Gamma}_{22}$ is sparse and lower triangular.

Conditional draws can be done at the cost of one unconditional draw and one conditional expectation. Suppose that the data are in the vector $Y_1$, and that $Y^* = (Y_1^*,Y_2^*)$ is an unconditional draw from $N(0,\Sigma)$. Then, conditionally on $Y_1$,
$-\widetilde{\Gamma}_{22}^{-1}\widetilde{\Gamma}_{21}(Y_1 - Y_1^*) + Y_2^*$
approximately has mean $\Sigma_{21}\Sigma_{11}^{-1}Y_1$ and covariance matrix $\Sigma_{22} - \Sigma_{21}\Sigma_{11}^{-1}\Sigma_{12}$ and thus is an approximate conditional draw of $Y_2$ given $Y_1$. We use ensembles of conditional draws in Section \ref{datasection} to quantify joint uncertainties in
interpolations of satellite data.

\subsection{Profile Likelihood with Linear Mean Parameters}

The Gaussian process model often includes a mean function that is linear with respect to a set of spatial covariates. \cite{stein2004approximating} outlined methods for computing residual (also known as restricted) likelihoods to avoid having to numerically maximize the approximate likelihood over the mean parameters. If one prefers to obtain the maximum approximate likelihood estimates of all parameters, it is usually computationally advantageous to profile out the mean parameters. Writing $E(Y) = X\beta$, where $X$ is the $n \times p$ design matrix, and $\beta$ is the vector of linear mean parameters, Vecchia's approximate likelihood becomes
\begin{align*}
\log p_{\theta,\beta,\tau,J}(y) = -\frac{n}{2}\log(2\pi) + \log \det \widetilde{\Gamma} - \frac{1}{2}\| \widetilde{\Gamma}( y - X\beta) \|^2.
\end{align*}
For fixed $\theta$, the maximum approximate likelihood estimate for $\beta$ is
\begin{align*}
\widehat{\beta}(\theta) = ( X^T \widetilde{\Gamma}^T \widetilde{\Gamma} X )^{-1} X^T \widetilde{\Gamma}^T \widetilde{\Gamma} y = \left[ (\widetilde{\Gamma}X)^T (\widetilde{\Gamma}X) \right]^{-1} (\widetilde{\Gamma}X)^T (\widetilde{\Gamma}y).
\end{align*}
Multiplying $\widetilde{\Gamma}X$ can be done column-by-column, each of which takes the same computational effort as multiplying $\widetilde{\Gamma}y$, which is the most computationally demanding task in the approximate likelihood evaluation. Thus profiling out the linear mean parameters can be done at the roughly the cost of $p$ additional approximate likelihood evaluations per iteration when maximizing over $\theta$.

\section{Numerical and timing comparisons}\label{numerical}

This section contains numerical results studying differences in the relative quality Vecchia's approximation with respect to choices for the permutation and grouping, under different model settings and in two, three, and four dimensional domains. Vecchia's approximation is also compared to a simple block independent Gaussian process approximation, since this approximation is easily constructed and has been shown to be competitive with state-of-the-art approximations \citep{stein2014limitations}. Comparisons are also made with stochastic partial differential equation (SPDE) approximations \citep{lindgren2011explicit} and tapered covariance approximations \citep{furrer2006covariance,kaufman2008covariance}. For all comparisons, the Mat\'ern covariance function is used,
\begin{align*}
M(r; \sigma^2, \alpha, \nu) = \frac{\sigma^2}{\Gamma(\nu)2^{\nu-1}}\left( \frac{r}{\alpha} \right)^\nu \mathcal{K}_\nu\left( \frac{r}{\alpha} \right),
\end{align*}
which has emerged as the model of choice for practitioners of spatial statistics \citep{guttorp2006studies}. The covariance function has three positive parameters $\sigma^2$, $\alpha$, and $\nu$, which are the variance, range, and smoothness parameters, respectively. The mean of the process is assumed to be known to be zero.

In the two-dimensional numerical studies, six different parameter choices are presented, with $\sigma^2 = 1$ in all six, while the range and smoothness cover all combinations of $\alpha \in \{ 0.1, 0.2 \}$ and $\nu \in \{1/2,1,3/2\}$. The locations form an $80\times 80$ regular grid on the unit square, giving 6400 locations. Example realizations from the six models are given in Figure \ref{simexamples} in Appendix \ref{appendixexamples}. Four orderings are considered: sorted coordinate, middle out, completely random, and AMMD. Ungrouped and grouped versions, and 30 and 60 nearest neighbors are considered. As mentioned in Section \ref{grouping}, calculations in the grouped versions effectively condition on more neighbors, so when discussing the results, ``grouped version with 30 neighbors'' refers to a grouped version of a likelihood approximation that initially conditioned on 30 neighbors. For each ordering and number of initial neighbors, Table \ref{group_statistics} includes statistics on the number of blocks $K$, the sizes of the blocks $(\# U_k)$, and the sizes of the conditioning sets $(\# \overline{J}_i)$.

\begin{table}
\centering
\begin{tabular}{ccccccc}
Ordering & $m$ & $K$ & mean $(\# U_k)$ & max $(\# U_k)$ & mean $(\# \overline{J}_i)$ & max $(\# \overline{J}_i)$ \\
\hline
MMD & 30 & 1406 & 51.35 & 124 & 55.83 & 123 \\
random & 30 & 1361 & 50.82 & 114 & 53.01 & 113 \\
coordinate & 30 & 788 & 58.95 & 171 & 66.02 & 170 \\
middle out & 30 & 773 & 60.33 & 154 & 60.35 & 153 \\
\hline
MMD & 60 & 754 & 116.46 & 291 & 133.36 & 290 \\
random & 60 & 714 & 118.46 & 308 & 129.74 & 307 \\
coordinate & 60 & 358 & 136.83 & 450 & 166.49 & 449 \\
middle out & 60 & 320 & 145.27 & 418 & 155.87 & 417
\end{tabular}
\caption{\label{group_statistics} For each ordering scheme and for both initial number of neighbor choices $m$, statistics on the number of blocks returned by the grouping algorithm $K$, the size of the neighbor sets in each group $(\# U_k)$, and the actual number of neighbors $(\# \overline{J}_i)$. Here $k = 1,\ldots,K$, and $i = 1,\ldots,n$.}
\end{table}

The combinations above consist of (6 parameter settings) $\times$ (4 orderings) $\times$ (30 vs.\ 60 neighbors) $\times$ (ungrouped vs.\ grouped) $=$ 96 settings, not counting the block approximations or the SPDE approximations. Exploring such a large number of scenarios with simulations would not be feasible since many replicates would be required to control Monte Carlo error size. Instead, we use two deterministic criteria to evaluate the different approximations: (1) KL-divergence from the implied approximate model to the target model, (2) asymptotic relative efficiency for estimating covariance parameters, computed using the the usual generalization of Fisher information for misspecified likelihoods \citep{heyde2008quasi}. If $\widetilde{\ell}(\theta)$ is the approximate likelihood, the information criterion is
\begin{align*}
H(\theta) =  E \Big( \nabla^2 \widetilde{\ell}(\theta) \Big) E\Big( \big(\nabla \widetilde{\ell}(\theta) \big) \big(\nabla \widetilde{\ell}(\theta) \big)^T \Big)^{-1} E \Big( \nabla^2 \widetilde{\ell}(\theta) \Big),
\end{align*}
where $\nabla$ indicates gradient, and $\nabla^2$ is the Hessian. The criterion $H(\theta)$ is of course equal to the Fisher information in the case that $\widetilde{\ell}$ is equal to the true loglikelihood. The inverse of $H$ gives the asymptotic covariance matrix for the three Vecchia likelihood parameter estimates, estimated simultaneously. Relative efficiencies are obtained by comparing the diagonal values of $H^{-1}$ to the diagonal values of the inverse of the Fisher information.

In the three- and four-dimensional numerical studies, the grids on the unit square are of size $19^3$ and $9^4$, giving totals of 6859 and 6561 locations. We use the exponential covariance function with range parameters $0.1, 0.2$, and $0.4$ Only the KL-divergence criterion and the ungrouped approximations are considered, but the same four orderings and up to 240 neighbors are used.

The section concludes with a timing study showing how the various computational tasks for Vecchia's approximation scale with the number of observations. All timing comparisons are done on a 2016 Macbook Pro with a 3.3GHz Intel Core i7 processor (two cores) with 16GB RAM. Code is mostly written in the R programming language \citep{Rcoreteam}, aside from the Vecchia likelihood functions, which have been implemented in C++. Computations are done while running R version 3.4.2, linked to Apple Accelerate multithreaded linear algebra libraries.

\subsection{Kullback-Leibler Divergence Comparisons}

For the $\nu=1$ cases, an SPDE approximation \citep{lindgren2011explicit} is also considered. The SPDE approximation provides computational savings by using a sparse inverse covariance matrix and is valid for integer values of the smoothness parameter in two spatial dimensions. To address edge effects in the SPDE approximation, we use the boundary extension described in \cite{lindgren2011explicit}. Several choices of boundary parameters were considered, with only small differences among choices. For the $\nu = 1/2$ case, we also consider a tapered covariance approximation \citep{furrer2006covariance,kaufman2008covariance} with three tapering ranges as a competitor. In contrast to Vecchia's approximation, the KL divergence from the SPDE or tapered approximations to the target model is not minimized at the target model's parameter values, so we perform an optimization to find the range and variance parameters to minimize the KL-divergence. SPDE computations were implemented using functions from the R-INLA package \citep{rue2009approximate}, tapering from the fields package \citep{fieldspackage}.

The KL divergence results are plotted in Figure \ref{kldiv}, and the discussion of the results is organized according to the smoothness parameter setting. When $\nu = 1/2$, which corresponds to an exponential covariance function, AMMD ordering reduces KL divergence relative to sorted coordinate ordering by factors of 16 and 22 for the two range parameter settings when the number of neighbors is 30. Both orderings run in roughly the same amount of time. When the observations are grouped, the approximations improve. When using the AMMD ordering, the improvement in KL divergence over the sorted coordinate ordering with no grouping increases to a factor of 64 when $\alpha = 0.1$ and a factor of 75 when $\alpha = 0.2$. For 60 neighbors the effect of ordering is more striking; AMMD ordering and grouping results reduces KL divergence by factors of 285 and 244.
%If we use 60 neighbors with the AMMD ordering and grouping, the computing time is only 1.29 times slower than ungrouped sorted coordinate ordering with 30 neighbors, but the approximation is an astonishing 1700 times more accurate for $\alpha = 0.1$.
Covariance tapering is not competitive for any of the tapering ranges. The grouped approximation with AMMD ordering and 30 neighbors is more than 12,000 times more accurate than the tapered approximation that runs in roughly the same amount of time.

\begin{figure}
\centering
\includegraphics[width=0.8\textwidth]{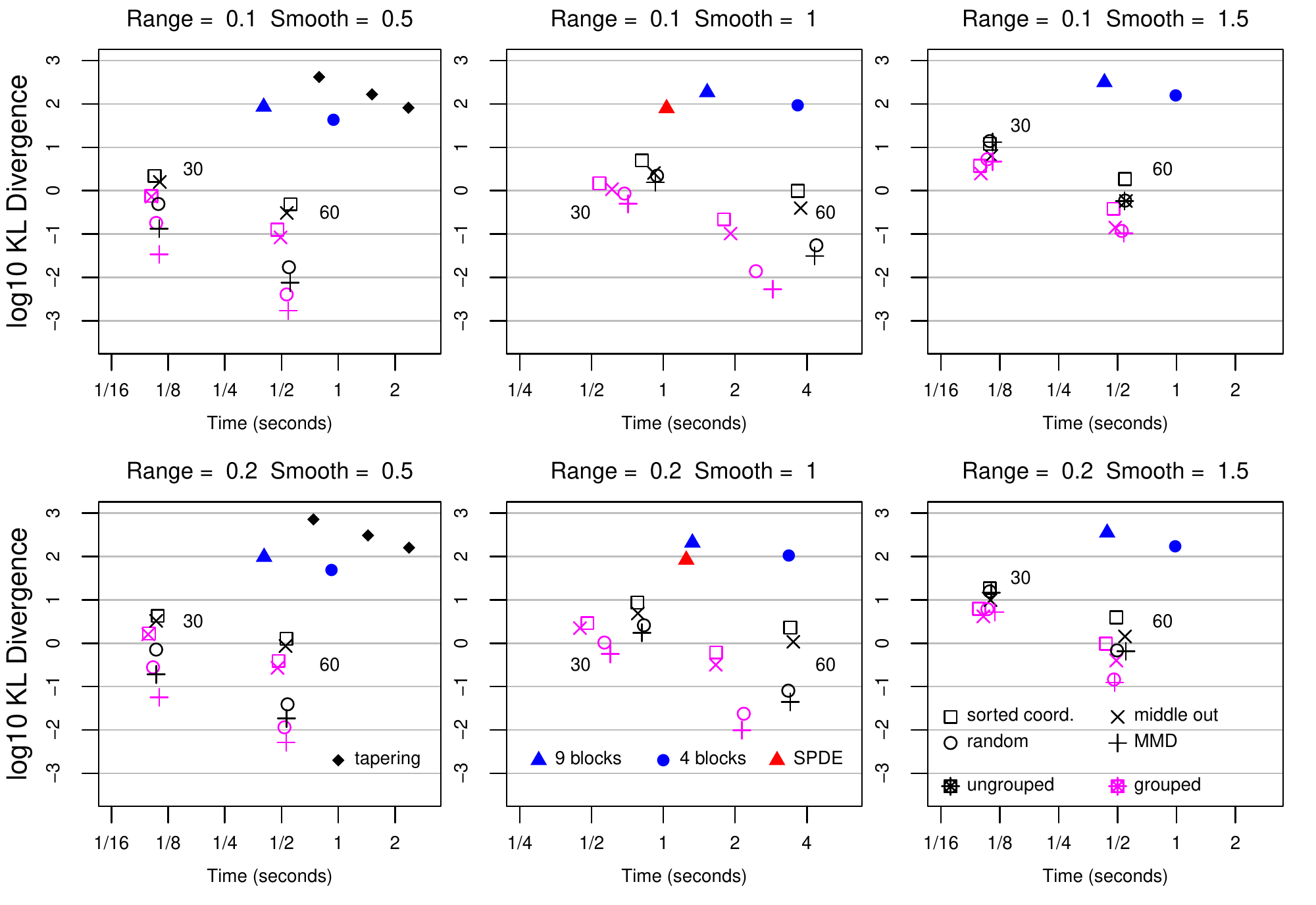}
\caption{{\small KL-divergences for Mat\'ern covariances with range $\in \{0.1,0.2\}$, and smoothness $\in \{1/2,1,3/2\}$. Locations form $80 \times 80$ grid on the unit square. Number of neighbors (30 and 60) indicated next to ungrouped symbols.}}
\label{kldiv}
\end{figure}

When $\nu = 1$, the major story is that the grouped approximations with AMMD ordering run faster than the SPDE approximation and are far more accurate; for $\alpha = 0.1$, KL-divergence is reduced by a factor of 160, and for $\alpha = 0.2$ by a factor of 148.
%Dramatic improvements in accuracy can be seen when using 60 neighbors, but at a computational cost. For example, the grouped AMMD ordering improves the approximation by a factor of 5700 over SPDE when $\alpha = 0.1$, and takes roughly 5 times longer to compute.
Computing times for the SPDE approximation include creating the mesh, forming the precision matrix (with boundary extension), and computing the sparse Cholesky factor. Reordering of the precision matrix is done by internal functions in the R Matrix package \citep{matrixpackage}. As before, the AMMD orderings are most accurate, especially with 60 neighbors.

When $\nu = 3/2$, which is the smoothest model considered, the AMMD ordering is again most accurate for 60 neighbors.  For 30 ungrouped neighbors, the differences among orderings is less substantial. In all cases, grouped approximations are more accurate than their ungrouped counterparts. In all parameter settings, the block independent approximations are never competitive in accuracy. The $3\times 3$ block approximation generally runs slower than the fastest grouped approximation and is orders of magnitude less accurate. Finer blocking in the block independent approximation will run faster but will decrease in accuracy.

KL-divergences for three and four dimensions are plotted in Figure \ref{kldiv3D4D}, where we see that the middle out ordering is the best choice in all but two covariance settings. In three dimensions, the random and MMD orderings appear to have KL-divergences that decay faster as the number of neighbors increase, whereas the sorted coordinate orderings have a slower decay. However, the gains for random and MMD come only after a large number of neighbors, which are unlikely to be used in practice due to the cubic scaling in number of neighbors. In four dimensions, middle out performs best. KL divergences were also computed for covariance tapering with three different taper ranges. Covariance tapering was not competitive with Vecchia's approximation in any setting; KL divergences were larger than $10^2$ in all settings, with some larger than $10^3$.

\begin{figure}
\centering
\includegraphics[width=0.8\textwidth]{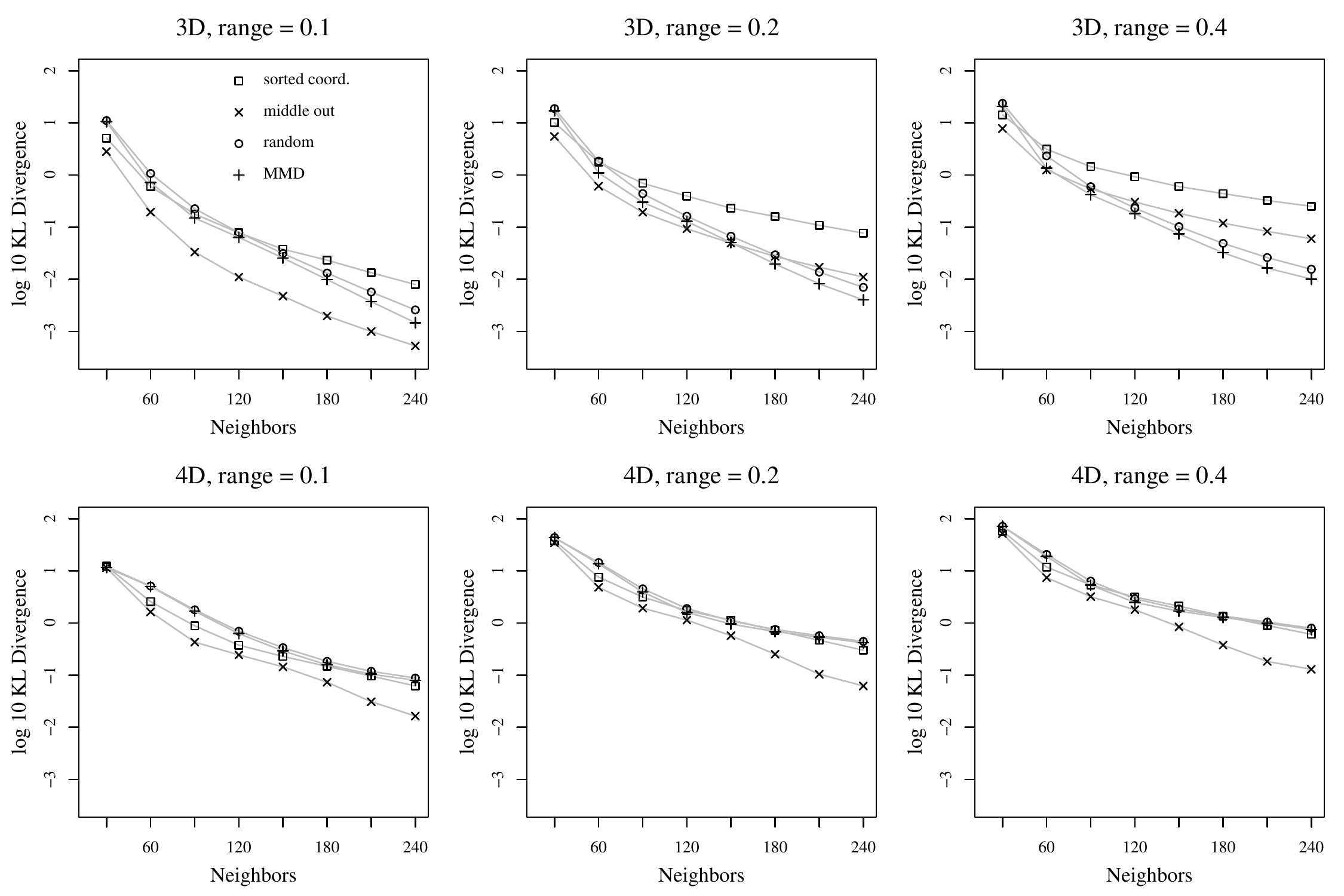}
\caption{{\small KL-divergences as a function of ordering and number of nearest neighbors for ungrouped version of Vecchia's approximation with exponential covariance. Locations form a regular grid in 3 and 4 dimensions, totaling 6859 and 6561 locations.}}
\label{kldiv3D4D}
\end{figure}

\subsection{Relative Efficiency Comparisons}

\begin{figure}
\centering
\includegraphics[width=0.8\textwidth]{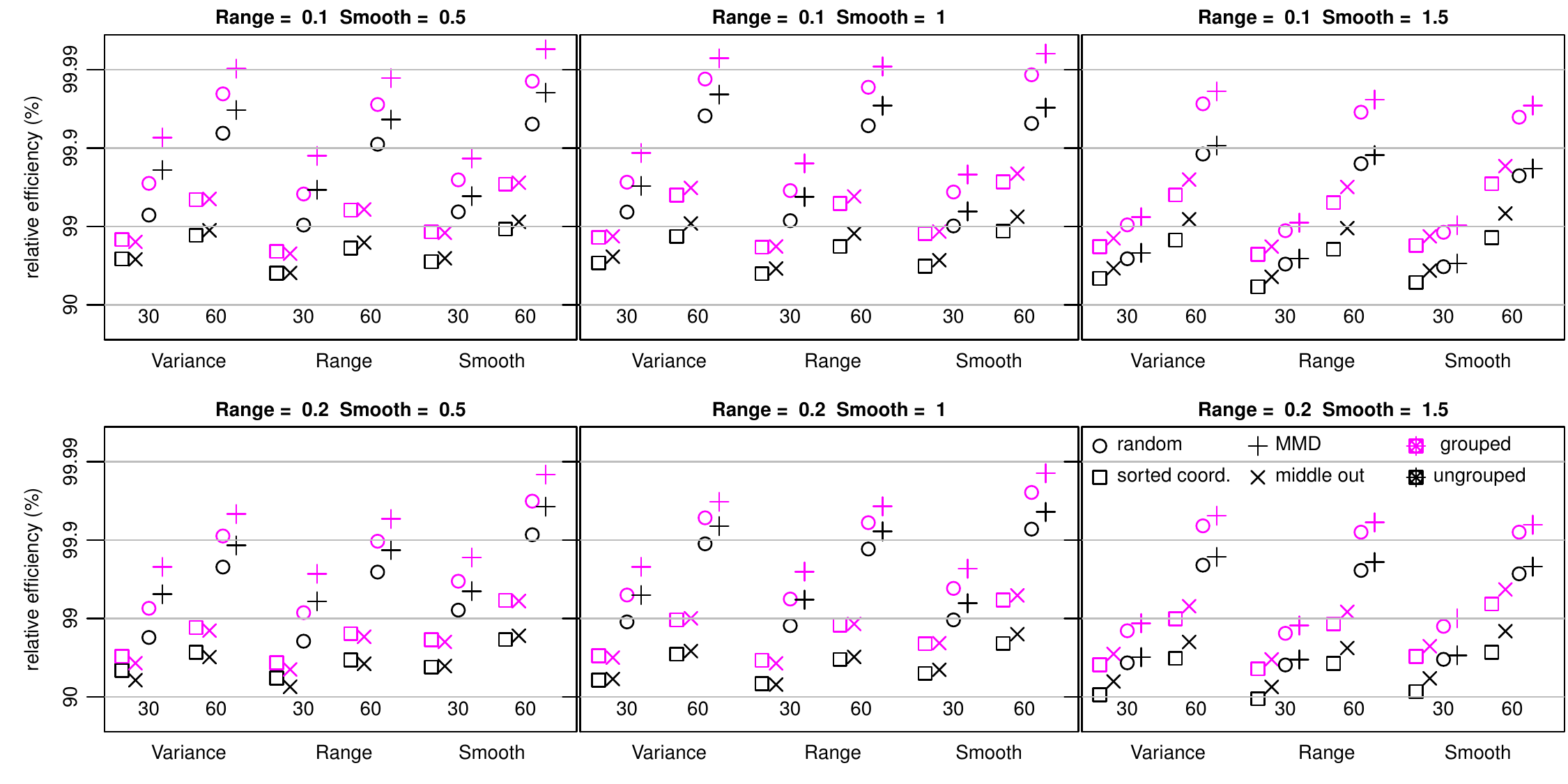}
\caption{{\small Relative efficiencies for estimating three Mat\'ern covariance parameters, variance $ =2$, range $\in \{0.1,0.2\}$, and smoothness $\in \{1/2,1,3/2\}$. Locations form $80 \times 80$ grid on the unit square, and the four orderings considered are sorted coordinate, middle out, completely random, and AMMD. Results for ungrouped (black) and not grouped (magenta) are presented.}}
\label{releff}
\end{figure}

The results of the numerical study on relative efficiency are summarized in Figure \ref{releff}. In every parameter setting, for all three parameters, and for both neighbor sizes, the AMMD ordering with grouping outperforms the default sorted coordinate ordering without grouping. In some cases, the difference is quite large. For example, when estimating the range parameter for the Mat\'ern with smoothness 1 and range 0.2, using the default sorted coordinate without grouping has 93.2\% relative efficiency, whereas AMMD with grouping has 99.7\% relative efficiency. This gain in efficiency comes at no additional computational cost.

In every case, using 30 neighbors with the AMMD ordering and grouping is superior to using 60 neighbors with the default sorted coordinate ordering without grouping. Since the computational complexity scales with the cube of the number of neighbors, this means our proposed improvements to Vecchia's likelihood can achieve increases in relative efficiency while simulataneously reducing the computational cost by a factor of 8. Finally, perhaps the most surprising result of this numerical study is that the default orderings are almost always outperformed by a completely random ordering of the points. This is also true in the KL divergence study. In both cases, the AMMD ordering offers further improvement.

\subsection{Timing}

In the context of analyzing data, it is useful to note which calculations must be carried out just once versus multiple times, and which calculations can be parallelized, although no explicit parallelization has been used in these studies. The ordering, grouping, and nearest neighbor calculations are computed just once. The likelihood calculations generally need to be repeated many times in the process of either numerically maximizing it with respect to covariance parameters, or sampling from posterior distributions in MCMC algorithms, but likelihood evaluations are embarassingly parallel due to the separability of Vecchia's specification. The nearest neighbor searches can be done in parallel.

Figure \ref{timing_increasing_n} presents the results of the timing study for an increasing number of observations with the exponential covariance function and AMMD ordering. The study is carried out for 30 and 60 neighbors and up to $10^5$ observations. For 30 neighbors, the slowest operation is grouping, followed by ordering, then finding neighbors, then the likelihood evaluations. The grouped likelihood for $10^5$ observations was computed in 1.5 seconds. For 60 neighbors, the grouping algorithm was slowest, followed by ordering. The grouped likelihood with 60 neighbors required 7.4 seconds for $10^5$ observations.

\begin{figure}
\centering
\includegraphics[width=0.7\textwidth]{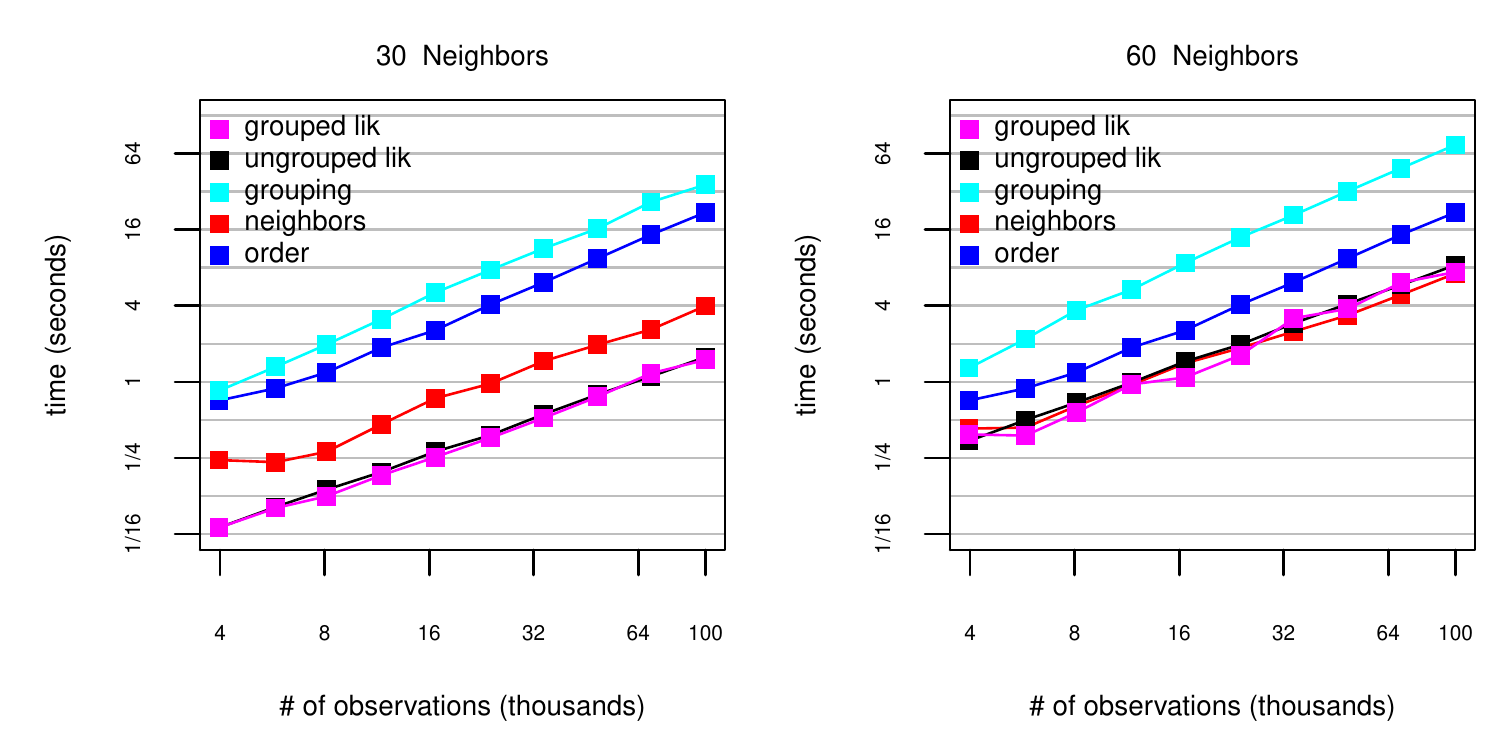}
\caption{{\small Timing Results for increasing number of observations regularly spaced on a square grid.}}
\label{timing_increasing_n}
\end{figure}

\section{Jason-3 Satellite Wind Speed Observations}\label{datasection}

Launched in January 2016, the Jason-3 satellite is the latest in a series of satellites equipped with radar altimeters for measuring ocean surface height, wave height, and ocean wind speed \citep{jason3website}. Jason-3 orbits the earth every 112 minutes along a path that repeats every 9.9 days. The data we consider are ocean surface wind speed values reported roughly once per second between August 4 and August 9, 2016 and are available at \url{http://www.nodc.noaa.gov/sog/jason/}. The goal of this analysis is to create interpolated maps of wind speed and quantify the joint uncertainty in the interpolations.

\textit{Data Preprocessing}: The Jason-3 data files include rain and ice flags to signify the presence of liquid water and ice along the radar signal's path. These disruptions degrade the quality of the signal, and the Jason-3 products handbook \citep{jason3products} states that flagged measurements should be ignored. As a conservative measure, we discard any measurement taken within 30 seconds of passing over rain or ice. In order to ensure that we can perform extensive comparisons among statistical analyses for data spanning a reasonably long time period, we average the one second measurements within 10 second intervals, discarding any intervals that have missing values. The resulting data vector has 18,973 values. A map of the data is plotted in Figure \ref{jason3data}.

\begin{figure}
\centering
\includegraphics[width=0.8\textwidth]{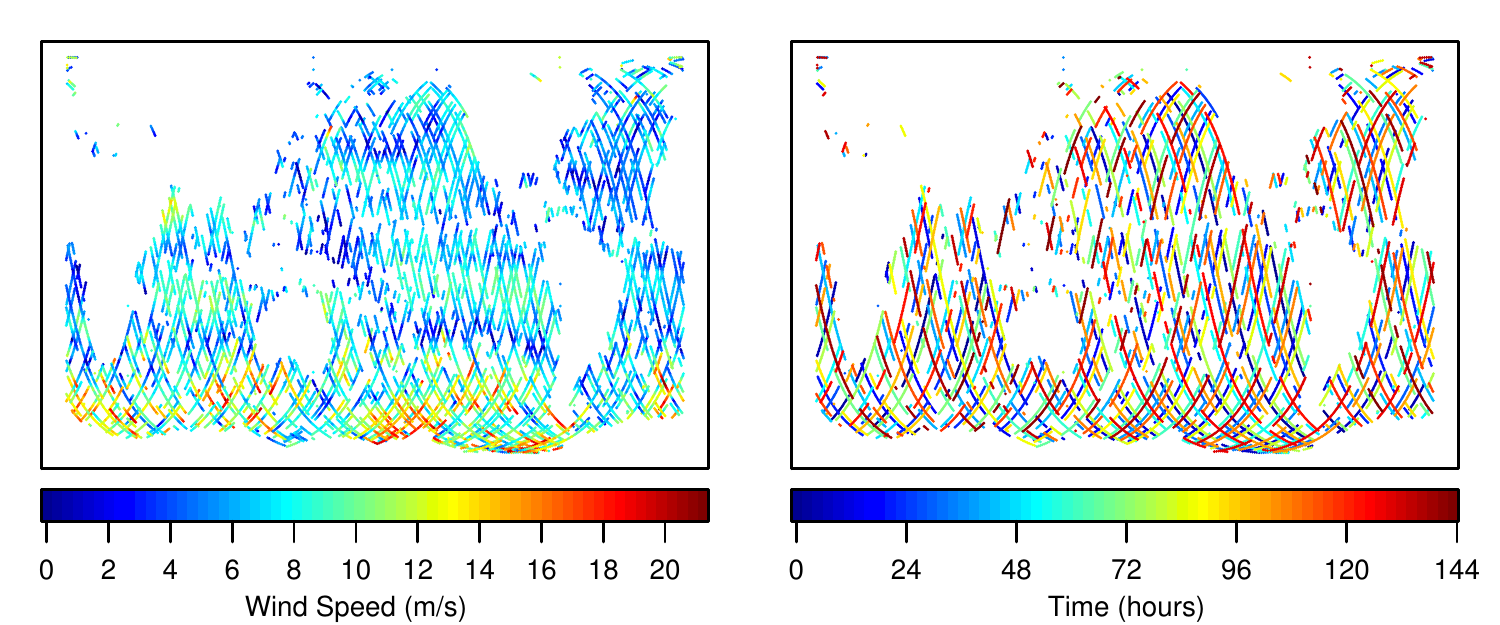}
\caption{{\small Jason-3 wind speed values and observation times.}}
\label{jason3data}
\end{figure}

Since the satellite can measure wind speeds at only a single location at any given time, it is expected that uncertainties in interpolated maps at specific times and locations will depend on whether there are nearby observations in space and time. This feature can be captured by modeling the data with a space-time Gaussian process. Specifically, for wind speeds at location $x \in \mathbb{S}^2$ and $t \in \mathbb{R}$, we consider a geostatistical model of the form
\begin{align*}
Y(x,t) = \mu + Z(x,t) + \varepsilon(x,t),
\end{align*}
where $\mu$ is considered nonrandom, $\varepsilon(x,t)$ is uncorrelated $N(0,\tau^2)$ (nugget term), and $Z(x,t)$ is a space-time Gaussian process with a Mat\'ern covariance function
\begin{align*}
\Cov(Z(x_1,t_1),Z(x_2,t_2))  = \frac{\sigma^2}{2^{\nu-1}\Gamma(\nu)}(d_{12})^\nu \mathcal{K}_\nu(d_{12}),
\end{align*}
where $\mathcal{K}_\nu$ is a modified Bessel function of the second kind, and
\begin{align*}
d_{12} = \sqrt{ \frac{ \| x_1 - x_2 \|^2 }{\alpha_1^2} + \frac{|t_1 - t_2|^2}{\alpha_2^2}},
\end{align*}
so that the covariance function is isotropic in space, stationary in time, but with different range parameters for the space and time dimensions. This method for constructing covariance functions on the sphere-time domain was originally used in \cite{jun2007approach}. The Mat\'ern covariance is not generally valid with great circle distance metric on the sphere \citep{gneiting2013strictly}, so Euclidean distance is used. \cite{porcu2016spatio} constructed some alternative space-time covariance functions, but \cite{guinness2016isotropic} argued that there is no reason to expect the use of covariance functions constructed on Euclidean spaces and restricted to the sphere to cause any distortions. The numerical results in \cite{porcu2016spatio} support this idea. The model we consider has five unknown covariance parameters, $(\sigma^2, \alpha_1, \alpha_2, \nu, \tau^2)$.

We consider four orderings for the observations in Vecchia's likelihood. The first is ordered in time, which I consider to be a default choice for this application. The second is completely random, the third is MMD in time, and the fourth is MMD in space. %ordering tailored to the structure of the space-time data, described as follows. I first construct a sequence $\{a_{ij}\}$ with $i=1,2,3,\ldots$, and $j=1,\ldots,2^{i-1}$, where $a_{ij} = (2j-1)/2^i$, giving
%\begin{align*}
%\frac{1}{2}, \hspace{3mm} \frac{1}{4}, \,\frac{3}{4}, \hspace{3mm} \frac{1}{8}, \,\frac{3}{8}, \,\frac{5}{8}, \,\frac{7}{8}, \hspace{3mm} \frac{1}{16}, \,\frac{3}{16}, \,\ldots
%\end{align*}
%multiply each member in the sequence by $n$, then round to the nearest integer, obtaining a sequence of integers $j_1, j_2, \ldots$. Then I use $j_1, j_2,\ldots$ as the permutation of the time ordering to obtain the third ordering.
Since the orbital pattern of the satellite follows a regular path over time, the MMD in time ordering provides good spatial coverage early in the ordering, as can be seen in Figure \ref{jason_first_1000}, which shows the locations of the first 1000 observations from the MMD in time ordering.

\begin{figure}
\centering
\includegraphics[width=0.5\textwidth]{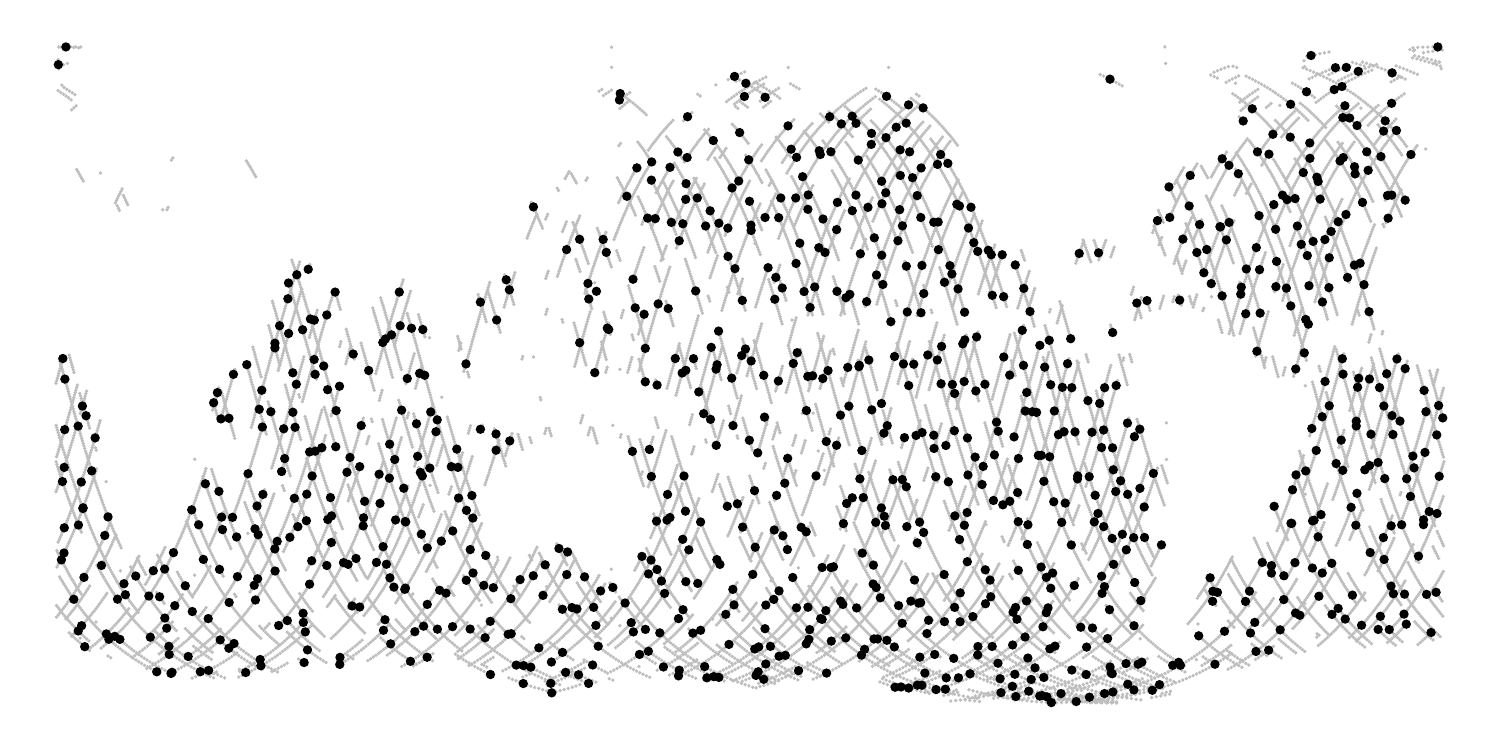}
\caption{{\small For MMD time ordering, first 1000 locations (black) and all locations (gray).}}
\label{jason_first_1000}
\end{figure}

For neighbor selection, \cite{datta2016nonseparable} defined ``nearest'' neighbors in space and time based on the correlation function. Here, we define distance based on the spatial dimension only. The reason for this choice is to encourage the neighbor sets to contain observations from different time points so that the conditional likelihoods contain information about the temporal range parameter.

Since the quality of Vecchia's approximation increases as neighbors are added, and maximum likelihood estimates are obtained in an iterative numerical search procedure, it is natural to consider a sequence of maximum approximate likelihood estimates, where each step in the sequence uses a larger number of neighbors than in the previous step. The estimates from an iterative procedure such as this can be monitored for convergence and stopped when successive parameter estimates do not change beyond some tolerance level. For each ordering, we start with 10 neighbors, find the maximum approximate likelihood estimate, and then we consider 20 neighbors, starting the optimization at the 10-neighbor parameter estimate, and so on up to 100 neighbors. As long as the optimization procedures are finding the legitimate maximum of the approximate likelihoods, the sequence of estimates from all orderings will all converge to the same maximum likelihood estimate, since all of the likelihood approximations converge to the exact likelihood when we condition on all possible past observations. Orderings can be compared to each other based on how many neighbors are required for convergence.

The results for the four orderings are plotted in Figure \ref{jason_parm_convergence}. Only the variance, spatial range, temporal range, and smoothness parameters are plotted because the nugget variance is esentially zero for every ordering and number of neighbors. It is clear that the parameter estimates from the random ordering and the two MMD orderings are converging more rapidly than are the estimates from the time ordering. The MMD orderings appear to be converging slightly more quickly than the random ordering. For the random and MMD orderings, the estimates for all parameters are within 2\% of the 100 neighbor estimates with just 50 neighbors, whereas the time ordered estimates do not appear to have settled down, even with 100 neighbors. The exception is the smoothness parameter, whose estimates for 10 neighbors are within 2\% of the 100 neighbor estimates for all orderings.

\begin{figure}
\centering
\includegraphics[width=0.8\textwidth]{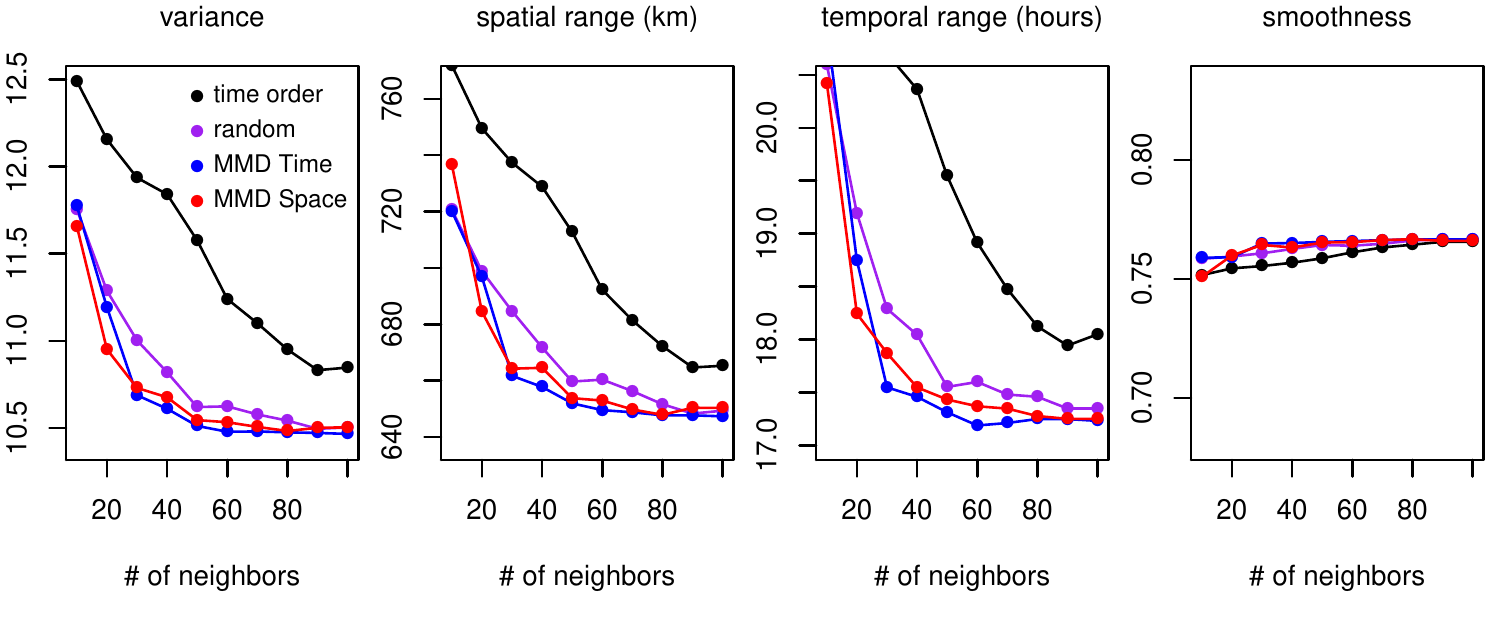}
\caption{{\small Estimated parameters for increasing number of neighbors. Vertical axis heights are 20\% of the value of the last estimate from ordering 3.}}
\label{jason_parm_convergence}
\end{figure}

Spatial interpolations are computed on an evenly spaced grid of size $120\times 240$ in latitude and longitude at two time periods. The first is at the average of all the observation times, to represent a hindcast--a prediction of past values from past data. The second is at the time of the last observation, to represent real-time interpolation of the data. Using ordering 3 for the observations, a random ordering for the prediction points, and 60 nearest neighbors with grouping, we compute 4000 conditional draws and use the sample variances of the conditional draws to estimate the conditional variances at each of the prediction points. The predictions and their standard deviations are plotted in Figure \ref{jason_pred}. As expected, the standard deviations are generally smaller for the hindcast predictions, and one can clearly see the effect of the orbital path of the satellite; locations along recently visited paths have smaller standard deviations. Further, the standard deviations are small at high and low latitudes, reflecting the fact that the satellite passes near the poles on every orbit. Finally, Figure \ref{jason_cond} contains two individual hindcast conditional draws and maps of conditional correlations with two points in the soutern Atlantic Ocean. One can see that the conditional correlation structure is quite complex and inhomogenous among the two points. A possible explanation is that this region has nearby observations in space and time, as can be seen from bottom left panel of Figure \ref{jason_pred}.

\section{Discussion}\label{discussion}

\begin{figure}
\centering
\includegraphics[width=0.8\textwidth]{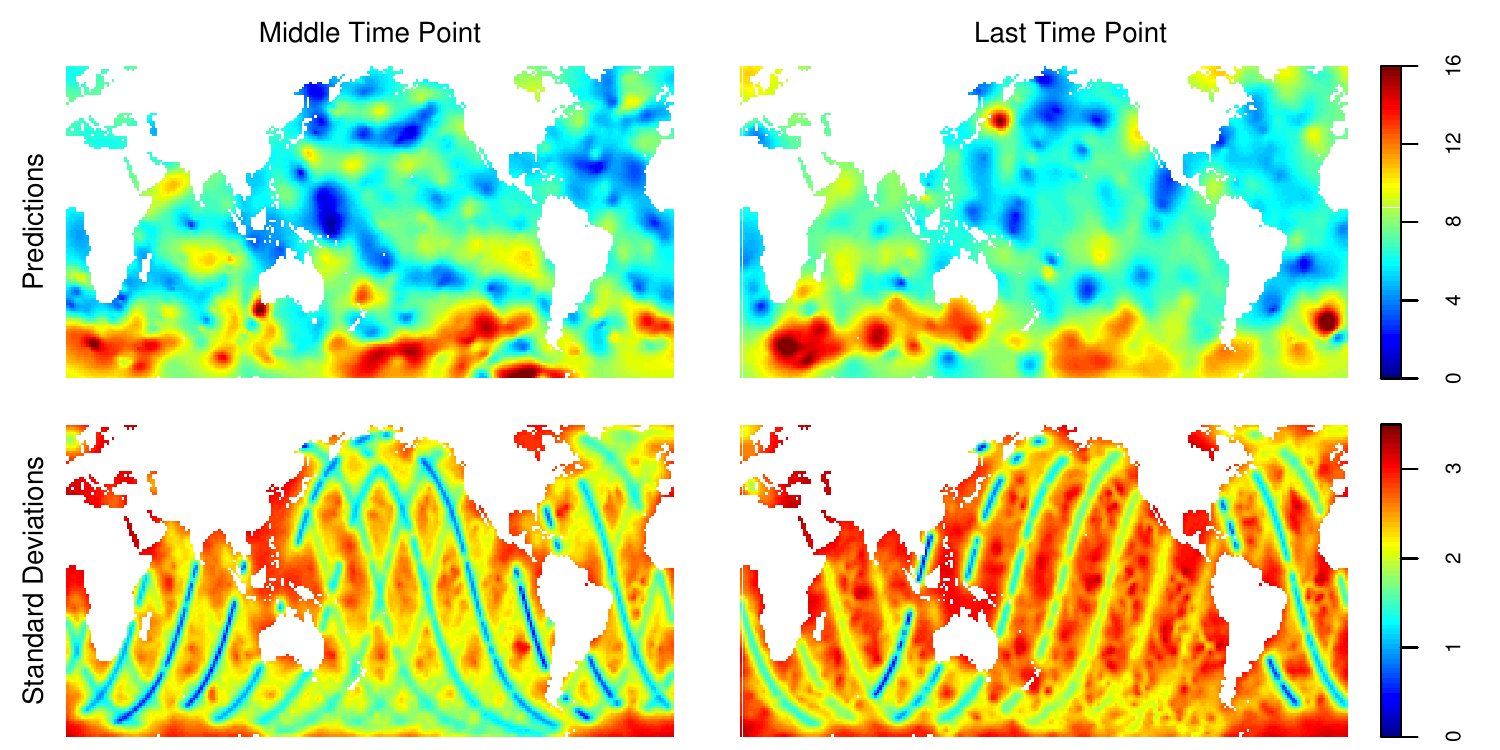}
\caption{{\small Predictions (top row) and simulated prediction standard deviations (bottom row) of windspeed at the mean observation time (left) and the last observation time (right).}}
\label{jason_pred}
\end{figure}

This paper demonstrates that reordering and grouping operations can lead to substantial improvements in the quality of Vecchia's approximation, sometimes by more than two orders of magnitude compared to default ungrouped approximations. Grouping also reduces the computational effort over ungrouped versions, and the grouping algorithm introduced here is general in that it can be applied to any choice of ordering and neighbor sets. I have also provided the R code written for reordering, finding nearest neighbors, automatic grouping, and likelihood evaluations.

One perhaps surprising result of this study is that Vecchia's approximation with MMD ordering and grouping with 30 neighbors runs faster than an SPDE approximation, while being two orders of magnitude more accurate in terms of KL divergence. Certainly more numerical results are needed in more cases, but this result coupled with the fact that Vecchia's approximation is valid for general covariance functions rather than being restricted to Mat\'ern with integer smoothness, suggests that Vecchia's approximation should be considered as a candidate for approximating Mat\'ern models whenever the SPDE approximation is considered.

The paper introduces a grouped version of Vecchia's approximation, based on a partitioning of the observations into blocks. Each block's contribution to the likelihood can be computed simultaneously. A theorem is provided to explain why grouping improves the model approximation, and an algorithm is described for finding a partition that is guaranteed control the memory burden. Exploring new partitioning algorithms for Vecchia's approximation could be a fruitful avenue for future work.

While developing some theory for grouping has been successful, a general theory for reordering remains elusive. One reason for this is the diversity of covariance functions and observation settings to be considered. Results on the screening effect \citep{stein2002screening} may be helpful for pushing the theory forward in some special cases, and perhaps theoretical results on estimates from subsamples could be useful as well \citep{hung2016variable}. However, I caution against drawing strong conclusions based on one-dimensional results. For example, for the exponential covariance in one-dimension, Vecchia's approximation is exact when observations are sorted along the coordinate and just a single nearest neighbor is used. Unsorted orderings require two neighbors for at least one of the points. I suspect that intuition gleaned from these one-dimensional examples has been incorrectly applied to the two-dimensional case, leading to the preference for sorted coordinate orderings in two dimensions. The numerical results in three and four dimensions are an attempt to study how ordering influences Vecchia's approximation in higher dimensions, which is relevant for computer experiment applications \citep{sacks1989design,santner2013design}. There is some evidence that KL divergence decays faster with the number of neighbors under random or MMD ordering in three dimensions, but middle out ordering appears to be best in four dimensions, at least in the examples considered here. This is further evidence to be wary of applying intuition from lower dimensions to higher dimensions.

This work considers the effect of ordering with the rule for choosing neighbor sets held constant. \cite{stein2004approximating} presented examples where including some distant neighbors can help in estimating parameters that control the behavior of the covariance function away from the origin. I have tried this approach but was not able to improve on nearest neighbors. A possible explanation is that I consider the Mat\'ern model with a different parameterization than in \cite{stein2004approximating}. I note here that the likelihood with MMD ordering automatically includes information about distant relationships since observations early in the ordering necessarily condition on distant observations since early observations are necessarily distant from one another.

Nonetheless, exploring the interaction between ordering and neighbor selection would be interesting. The work on the local approximate Gaussian process (laGP)  \citep{gramacy2015local} explores fast automatic neighbor selection in a different context and could be relavent to this question. While neighbors figure prominently in both the present paper and in laGP, neighbors in laGP can be selected from anywhere in the ordering, whereas Vecchia's approximation requires neighbors to come from previous in the ordering. Ordered neighbors ensure that the approximation corresponds to a valid joint density, which is intended to approximate a specified joint density, whereras laGP constructs a framework for interpolation and does not target any particular global model.

\begin{figure}
\centering
\includegraphics[width=0.8\textwidth]{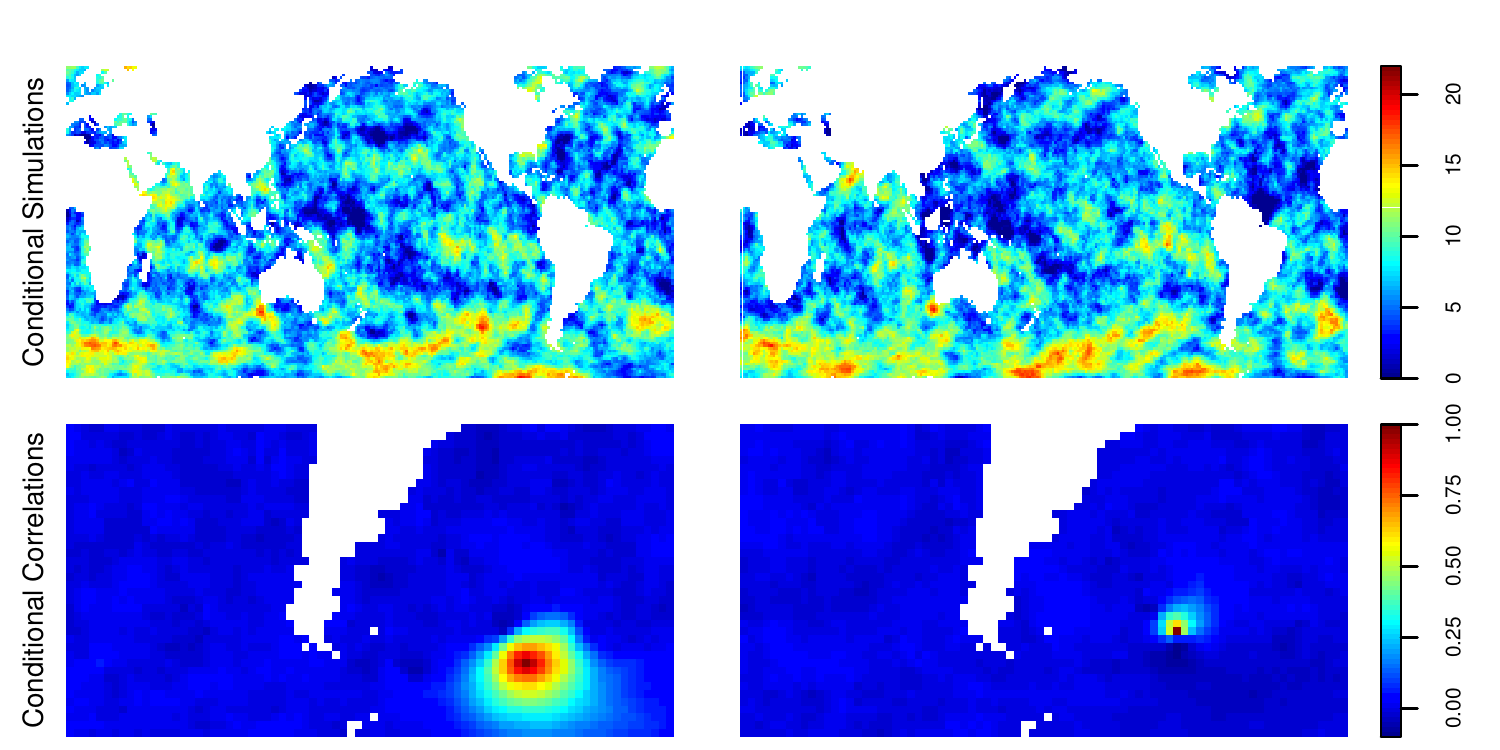}
\caption{{\small For hindcast, two individual conditional draws (top row) and empirical conditional correlations at two points in the southern Atlantic.}}
\label{jason_cond}
\end{figure}

\vskip24pt

\begin{center}{\Large \textbf{Acknowledgements} }\end{center}

This material is based upon work supported by the National Science Foundation under Grant Number 1613219.

\vskip24pt

\bibliographystyle{apa}
\bibliography{refs}

\appendix

\section{Realizations from Models Studied}\label{appendixexamples}

Figure \ref{simexamples} contains example realizations from the models used in the KL-divergence and relative efficiency studies.

\begin{figure}[h]
\centering
\includegraphics[width=0.7\textwidth]{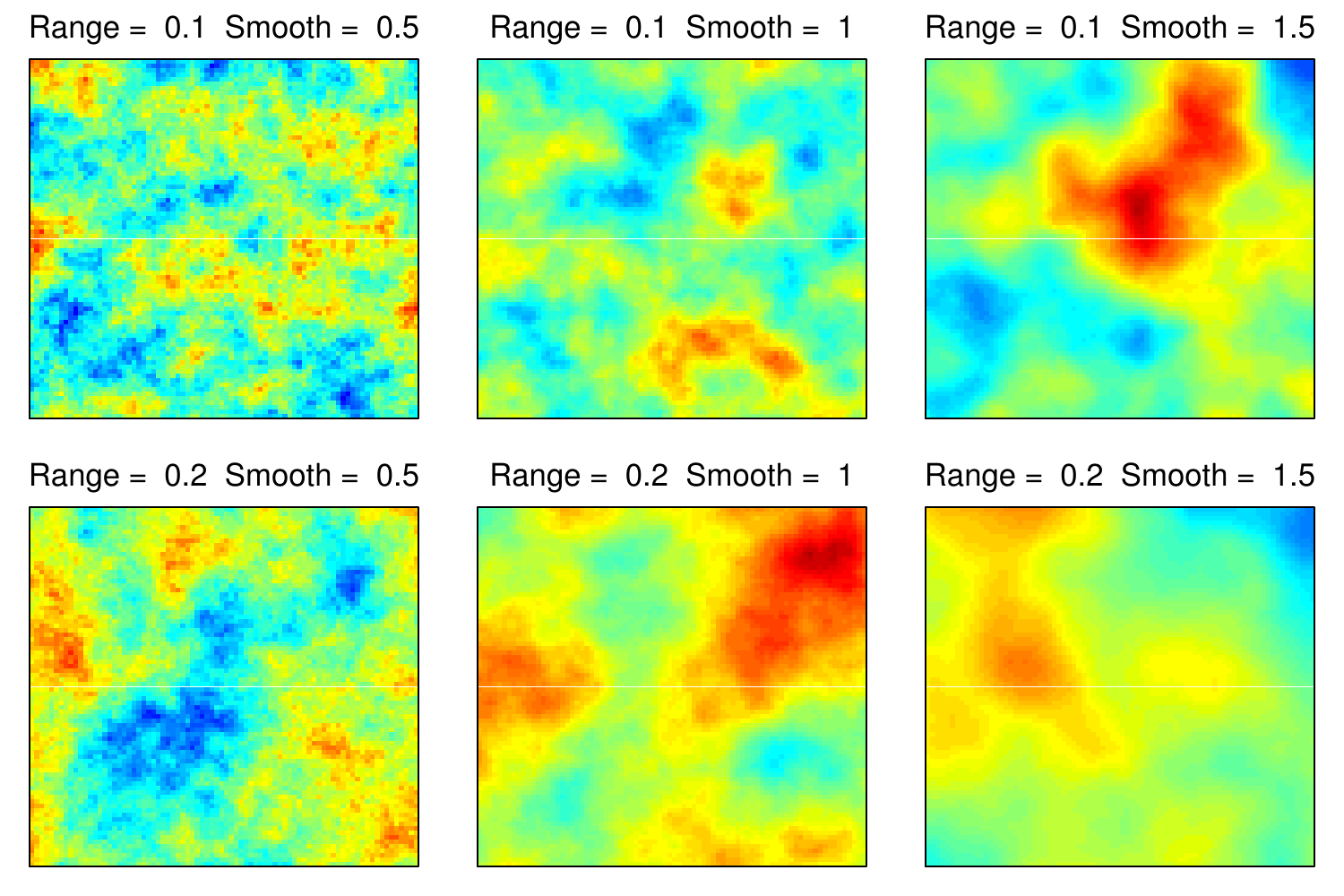}
\caption{{\small Gaussian process realizations at the six \matern\ parameter settings.}}
\label{simexamples}
\end{figure}

\end{document}